\documentclass[twocolumn,smallcondensed,smallextended]{svjour3}
\usepackage{fix-cm}
\usepackage{bbm}
\usepackage{amsmath}
\usepackage{amssymb}
\usepackage{latexsym}
\usepackage{graphics}
\journalname{Quantum Information Processing}
\begin{document}

\title{
Tractable measure of nonclassical correlation
using density matrix truncations
}
\titlerunning{Tractable measure of nonclassical correlation}
\author{Akira SaiToh \and Robabeh Rahimi \and Mikio Nakahara}
\institute{A. SaiToh \at
Research Center for Quantum Computing,
Interdisciplinary Graduate School of Science and Engineering,
Kinki University, 3-4-1 Kowakae, Higashi-Osaka, Osaka 577-8502, Japan\\
Tel.: +81-6-6721-2332\\
Fax: +81-6-6727-4301\\
\email{saitoh@alice.math.kindai.ac.jp}
\and
R. Rahimi \at
Departments of Chemistry and Materials
Science, Graduate School of Science, Osaka City University,
3-3-138 Sugimoto, Sumiyoshi-ku, Osaka 558-8585, Japan\\
{\it Note:} She belonged to Department of Physics, Kinki
University when this work was mostly conducted.\\
\email{rahimi@sci.osaka-cu.ac.jp}
\and
M. Nakahara \at
Research Center for Quantum Computing,
Interdisciplinary Graduate School of Science and Engineering,
Kinki University, 3-4-1 Kowakae, Higashi-Osaka, Osaka 577-8502, Japan;\\
Department of Physics, Kinki University, 3-4-1 Kowakae,
Higashi-Osaka, Osaka 577-8502, Japan\\
\email{nakahara@math.kindai.ac.jp}
}

\date{Received: date / Accepted: date}

\maketitle

\begin{abstract}
In the context of the Oppenheim-Horodecki paradigm of nonclassical
correlation, a bipartite quantum state is (properly) classically
correlated if and only if it is represented by a density matrix having a
product eigenbasis. On the basis of this paradigm, we propose a measure
of nonclassical correlation by using truncations of a density matrix
down to individual eigenspaces. It is computable within polynomial time
in the dimension of the Hilbert space albeit imperfect in the detection
range. This is in contrast to the measures conventionally used for the
paradigm. The computational complexity and mathematical properties of
the proposed measure are investigated in detail and the physical picture
of its definition is discussed.
\end{abstract}



\section{Introduction}\label{sec1}
Classical/nonclassical separation of correlations between
subsystems of a bipartite quantum system has been an essential
and insightful subject in quantum information theory.
The entanglement paradigm \cite{W89,P96} is based on the state preparation
stage: any quantum state that cannot be prepared by local operations and
classical communications (LOCC) \cite{PV06} is entangled.
There are paradigms \cite{B99,Z02,O02,GPW05} based on post-preparation
stages, which use different definitions of classical and nonclassical
correlations.
On the basis of the Oppenheim-Horodecki definition \cite{O02,H05,PS08},
a quantum bipartite system consisting of subsystems A and B
is (properly) classically correlated if and only if it is described
by a density matrix having a product eigenbasis (PE),\\
$
 \rho^{\rm AB}_{\rm PE}=\sum_{j,k=1,1}^{d^{\rm A},d^{\rm B}}
e_{jk}|v_j^{\rm A}\rangle\langle v_j^{\rm A}|
\otimes |v_k^{\rm B}\rangle\langle v_k^{\rm B}|,
$
where $d^{\rm A}$ ($d^{\rm B}$) is the dimension of the Hilbert space of
 A (B), $e_{jk}$ is the eigenvalue of $\rho^{\rm AB}_{\rm PE}$
corresponding to an eigenvector
$|v_j^{\rm A}\rangle\otimes |v_k^{\rm B}\rangle$.
Thus, a quantum bipartite system consisting of subsystems A and B
is nonclassically correlated if and only if it is described by a density
matrix having no product eigenbasis. In this Paper, we employ
this classical/nonclassical separation.

This definition was introduced in the discussions by Oppenheim {\em et
al.} \cite{O02,H05} on information that can be localized by applying
closed LOCC (CLOCC, a branch family of LOCC) operations. The
CLOCC protocol allows only local unitary operations and the operations of
sending subsystems through a complete dephasing channel.
The classical/nonclassical separation is linked to a localizable
information under the zero-way CLOCC protocol in which
coherent terms are deleted completely by local players before
communicating under CLOCC. A bipartite state with a product
eigenbasis carries information completely localizable under zero-way
CLOCC.
The nonlocalizable information under zero-way CLOCC is a measure of
nonclassical correlation. To bring physical insights to the paradigm,
the operational interpretations of nonclassical correlation have been
investigated in view of impossibility of local broadcasting \cite{PHH08}
and in view of nonclassicality in local measurements \cite{H05,L08}.
Furthermore, a correspondence between a separable state and a
classically correlated state in the context of a system extension by
using local ancillary systems was reported \cite{LL08}. The paradigm is
thus physically meaningful, for which it is of growing interest
to introduce measures \cite{Z02,O02,H05,G07,SRN08,SRN08-2,PHH08} and
detection methods \cite{F10,RS10,BC10} (a measure can, of course, be
regarded as a detection tool). 

To define measures of nonclassical correlation, there must be
certain axioms or requirements to satisfy. To find proper axioms for them,
let us revisit some axioms \cite{PV06} for entanglement measures.
(Although we do not revisit, one may also be interested in axioms for
classical correlation measures \cite{HV01}.)
The commonly accepted axioms for entanglement measures are 
(i) a measure should vanish for any separable state, and (ii)
entanglement and its measure should not increase under LOCC. These are
regarded as the most important axioms. In addition, one must
have certain maximally entangled states for which a measure takes
its maximum value as a consequence of (ii).
Similar axioms are naturally to be imposed to measures of nonclassical
correlation. The following axioms are suggested in the present context where
the basic protocol is the zero-way CLOCC.
(i') a measure should vanish for any state having a product eigenbasis,
and (ii') a measure should be invariant under local unitary
transformations. It should be noted that one cannot require a measure to
be nonincreasing under general local operations.\footnote{
There are local operations that increase nonclassical correlation.
For example, consider a 3-qubit state of the system ${\rm AB}$,
$(|01\rangle^{\rm A}\langle 01|\otimes |0\rangle^{\rm B}\langle 0|
+ |+0\rangle^{\rm A}\langle +0|\otimes |1\rangle^{\rm B}\langle 1|)/2$,
with $|+\rangle =(|0\rangle+|1\rangle)/\sqrt{2}$. This density matrix
has a product eigenbasis and hence the state possesses only classical
correlation. Let us apply a complete projection with the projectors
$|+\rangle\langle +|$ and $|-\rangle\langle -|$ to the second qubit.
The resultant state is represented as a density matrix that has the
eigenvectors $|0+\rangle^{\rm A}|0\rangle^{\rm B}$,
$|0-\rangle^{\rm A}|0\rangle^{\rm B}$,
$|++\rangle^{\rm A}|1\rangle^{\rm B}$, and
$|+-\rangle^{\rm A}|1\rangle^{\rm B}$ with the equal populations.
This has no product eigenbasis and hence the state possesses
nonclassical correlation.}
One may also need a measure to take the maximum value
for certain maximally entangled states since they should possess quite large
nonclassical correlation.

The above axioms (i') and (ii') are satisfied by the nonlocalizable information
named zero-way quantum deficit \cite{O02,H05} and another well-known
measure of nonclassical correlation, the minimized quantum discord \cite{Z02}.
They take the maximum values (dependent on the system size) for the generalized
Bell state $|\Psi_{\rm Bell}\rangle \langle\Psi_{\rm Bell}|$ with
$|\Psi_{\rm Bell}\rangle = \sum_{i=1}^{N}|i\rangle|i\rangle/\sqrt{N}$ for an
$N\times N$-dimensional bipartite system.
The quantum discord is a discrepancy between two different expressions
of mutual informations that are identical in the classical probability
theory. Its minimization over the choices of a local basis is local-unitary
invariant and can be used as a measure of nonclassical correlation in the
present context. Recently, the minimized one is simply referred to as
quantum discord. We follow this custom hereafter. Quantum discord for one
side, say, the subsystem ${\rm A}$, as an ``apparatus'' vanishes if and only
if \cite{Z02} the bipartite density matrix is written as
$\sum_k |k\rangle^{\rm A} \langle k|\otimes \xi_k^{\rm B}$
with $\xi_k^{\rm B}$ positive Hermitian operators acting on the other
side ${\rm B}$. Thus the average of the quantum discords for both sides
is a measure with the full detection range of nonclassical correlation.

There are other measures \cite{G07,SRN08,SRN08-2,PHH08} that were later
proposed on the basis of the same definition of classical/nonclassical
correlations. In particular, Piani {\em et al.} \cite{PHH08} designed a
measure which vanishes if and only if a state has a product eigenbasis.
It is in a similar form as quantum discord and defined as a
distance of two different quantum mutual informations that is minimized
over local maps associated with local positive operator-valued
measurements \cite{D78}. A problem in a practical point of view is that
the original nonlocalizable information, quantum discord, and the
Piani {\em et al.}'s measure all require expensive computational tasks to
find minima over all possible local operations in their contexts.
A similar difficulty exists in the measurement-induced disturbance
\cite{G07,L08,D09} for which a minimization is actually required to find
a proper Schmidt basis (or local dephasing basis) used to compute the measure
in degenerate cases (namely, the cases where the eigenbases of one or
both of the reduced density matrices of the state are not unique).
Note that detecting nonclassical correlation can be easier than
quantification; apart from measures, there are easy cases for detection as
we later discuss in Section~\ref{sec4}. Our interest is, however, not only
detecting nonclassical correlation but quantification.

In our previous work \cite{SRN08}, an entropic measure $G$
based on a sort of a game to find the eigenvalues of a reduced density matrix
from the eigenvalues of an original density matrix was proposed.
This measure can be computed within a finite time although it does not
have a perfect detection range. Its computational cost is exponential in
the dimension of the Hilbert space. One way \cite{ourarxivEnCE} to
achieve a polynomial cost is to introduce carefully-chosen maps similar to 
positive-but-not-completely-positive maps \cite{P96,H96}. We pursue a
different way in this paper.

Here, we introduce a measure of nonclassical correlation for a bipartite
state using the eigenvalues of reduced matrices obtained by tracing out
a subspace after certain truncations of a density matrix. Its
construction is rather simple as we see in Definition~\ref{defM} of
Section~\ref{sec3}. The computational cost is shown to be polynomial in
the dimension of the Hilbert space. Although the measure is imperfect in
the detection range and possesses no additivity property, it is
practically useful as an economical measure invariant under local
unitary operations. It takes the maximum value for the generalized Bell
states. In addition, it reduces to the entropy of entanglement (see,
e.g., \cite{BZ06} for the definition) for pure states with the Schmidt
coefficients $\le 1/\sqrt{2}$. (Here, a Schmidt coefficient is a square
root of a non-zero eigenvalue of a reduced density matrix of a subsystem
for a given pure state.)

This paper is organized as follows. We begin with a brief overview of the
measure $G$ in Section~\ref{sec2}. The measure $M$ is introduced and its
properties are investigated in Section~\ref{sec3}.  The validity of $M$
as a measure and a physical interpretation of its definition are
discussed in Section~\ref{sec4}. Section~\ref{sec5} summarizes this work.

\section{Brief overview of the measure by partitioning eigenvalues}\label{sec2}
We first make a brief overview of the measure $G$, an existing measure
computable in finite time.
In the context of bipartite splitting, it is defined as the minimized
discrepancy between the set of the mimicked eigenvalues of a local
system (say, subsystem A),
$\{\tilde{e_i}\}_{i=1}^{d^{\rm A}}$, and the set of the genuine eigenvalues
of the local system, $\{e_i\}_{i=1}^{d^{\rm A}}$.
Here, $\tilde{e_i}$'s are calculated by
(i) partitioning the $d^{\rm A}\times d^{\rm B}$ eigenvalues of the
original bipartite state $\rho^{\rm AB}$ into $d^{\rm A}$ sets; and
(ii) calculating the sum of the $d^{\rm B}$ elements in each set.
The discrepancy in view from one side (from Alice's side in this
context) is defined as 
\[
  F^{\rm A}(\rho^{\rm AB})=
\underset{\mathrm{partitionings}}{\mathrm{min}}\left|
\sum_{i} (\tilde e_i\log_2 \tilde e_i - e_i\log_2 e_i)
\right|.
\]
Similarly $F^{\rm B}(\rho^{\rm AB})$ is defined.
The measure is defined as
\[
G(\rho^{\rm AB})=\mathrm{max}[F^{\rm A}(\rho^{\rm AB}),
F^{\rm B}(\rho^{\rm AB})].
\]

A drawback of the measure is that
the number of combinations of eigenvalues that should be tried in the
minimization is ${}_{d^{\rm A}d^{\rm B}}{\rm C}_{d^{\rm B}}
\times{}_{(d^{\rm A}-1)d^{\rm B}}{\rm C}_{d^{\rm B}}\times~\cdots~\\
\times{}_{[d^{\rm A}-(d^{\rm A}-1)]d^{\rm B}}{\rm C}_{d^{\rm B}}
=(d^{\rm A}d^{\rm B})!/(d^{\rm B}!)^{\rm d^{\rm A}}
\simeq 2^{d^{\rm A}d^{\rm B}\log_2 d^{\rm A}}$ when the
subsystem of concern is ${\rm A}$ [$(d^{\rm A}d^{\rm B})!/
(d^{\rm A}!)^{\rm d^{\rm B}}\simeq
2^{d^{\rm A}d^{\rm B}\log_2 d^{\rm B}}$ when it is ${\rm B}$].
Indeed, this complexity is better in practice than that for
minimization over all certain local operations required for
calculating zero-way quantum deficit, quantum discord, and Piani {\em
et al.}'s measure. The complexity of a minimization over all local
operations for a subsystem, say ${\rm A}$, is
$O[{\rm poly}(d^{\rm A},d^{\rm B})\times 2^{(d^A)^2\log_2 c}]$ with $c$ the number of
values tried for each parameter of a local operation. The complexity
for computing $G$ is smaller in the range $d^{\rm A},d^{\rm B}\lesssim c$.
Computing $G$ is, however, still very expensive.

\section{Measure based on partial traces of truncated density matrices}\label{sec3}
A measure that is computable within realistic time is desired for
practical use. We introduce in the following a measure
that achieves a realistic computational time, namely, polynomial time
in the dimension of the Hilbert space.
\subsection{Introduction of the measure}
Let us begin with a basic definition.
\begin{definition}\label{def1}
Let us write the eigenspace corresponding to the eigenvalue $\eta$ of
a bipartite density matrix $\rho^{\rm AB}$ as
${\rm span}\{|v_k^\eta\rangle \}_{k=1}^{d^{\eta}}$
where $d^{\eta}$ is the dimension of the eigenspace and
$|v_k^\eta\rangle$'s are the eigenvectors.
Let us define a ``truncated'' density matrix down to the $\eta$ eigenspace
as
\[
\tilde{\rho}^{\eta}=\eta\sum_{k=1}^{d^{\eta}}
|v_k^\eta\rangle\langle v_k^\eta|.
\]
\end{definition}
The following proposition holds for $\tilde{\rho}^{\eta}$.
\begin{proposition}\label{prop1}
Consider $\tilde{\rho}^{\eta}$ introduced above.
The eigenvalues of the reduced matrix
${\rm Tr}_{\rm B}\tilde{\rho}^{\eta}$ (${\rm Tr}_{\rm A}\tilde{\rho}^{\eta}$)
of the system ${\rm A}$ (${\rm B}$) are integer multiples of $\eta$
if $\rho^{\rm AB}$ has a product eigenbasis.
\end{proposition}
\begin{proof}
Suppose $\rho^{\rm AB}$ has a product eigenbasis
$\{|a_i\rangle\}_{i=1}^{d^{\rm A}}\times \{|b_j\rangle\}_{j=1}^{d^{\rm B}}$.
Then $\tilde{\rho}^\eta$ becomes
$\eta\sum_{k=1}^{d^{\eta}} |a_k^\eta \rangle\langle a_k^\eta|
\otimes |b_k^\eta\rangle\langle b_k^\eta|$ where $|a_k^\eta\rangle$'s
($|b_k^\eta\rangle$'s) are some $d^{\eta}$ vectors, with possible
multiplicities, found in
$\{|a_i\rangle\}_{i=1}^{d^{\rm A}}$ ($\{|b_j\rangle\}_{j=1}^{d^{\rm B}}$).
Since $|a_i\rangle$'s ($|b_j\rangle$'s) are orthogonal to each other,
it is now easy to find that the proposition holds.
$\Box$\end{proof}
\begin{remark}
The sum of the eigenvalues of ${\rm Tr}_{\rm B}\tilde{\rho}^{\eta}$
is equal to ${\rm Tr}\tilde{\rho}^{\eta}=\eta d^{\eta}$; similarly,
that of ${\rm Tr}_{\rm A}\tilde{\rho}^{\eta}$ is equal to $\eta d^{\eta}$.
This property is tacitly used in the calculations hereafter.
\end{remark}
Let us introduce useful functions.
\begin{definition}
For $x,y\ge 0$, let us introduce the function ${\rm nim}_y(x)$
which is the nearest integer multiple of $y$ for $x$.
For the exceptional case that $x\pm y/2$ are integer multiples
of $y$, let it take the value $x-y/2$. Thus, the strict definition is
\[
 {\rm nim}_y(x) =\left\{
\begin{array}{ll}
y\lfloor x/y \rfloor~~~& (x-y\lfloor x/y \rfloor \le y\lceil x/y\rceil-x, ~y\not = 0)\\
y\lceil  x/y \rceil~~~&  (x-y\lfloor x/y \rfloor > y\lceil x/y\rceil-x, ~y \not = 0)\\
0~~~& (y=0)
\end{array}\right..
\]
\end{definition}
\begin{definition}
Consider the two collections of $m$ nonempty sets,
\[
X = \left\{
\{x_i^1\}_{i=1}^{d_1},\{x_i^2\}_{i=1}^{d_2},\ldots,\{x_i^m\}_{i=1}^{d_m}
\right\}
\]
and
\[
Y = \left\{
\{y_i^1\}_{i=1}^{d_1},\{y_i^2\}_{i=1}^{d_2},\ldots,\{y_i^m\}_{i=1}^{d_m}
\right\}
\]
with nonnegative real numbers $x_i^j$ and $y_i^j$ ($j=1,\ldots,m$) such
that $\sum_{j=1}^m\sum_{i=1}^{d_j}x_i^j
=\sum_{j=1}^m\sum_{i=1}^{d_j}y_i^j= 1$ where
$d_j$ is the size of the $j$th set ($d_j$ is common for $X$ and $Y$ for
the same $j$). Here, $Y$ is assumed to be a prediction or an estimate of $X$.

Let us here introduce the quota 
\[
T_j=\sum_{i=1}^{d_j} x_i^j
\]
for the $j$th set. The quotas satisfy $\sum_{j=1}^m T_j = 1$.

As a discrepancy between $x_i^j$ and $y_i^j$,
we may use the quantity
\[
 s(x_i^j,y_i^j)=-\left|x_i^j - y_i^j\right|\log_2 (x_i^j/T_j).
\]
This can be interpreted as the claim that we obtain
$[(-\log_2 x_i^j) - (-\log_2 T_j)]$-bit of information for the $(i,j)$th
event with the weight $\left|x_i^j - y_i^j\right|$ of surprise when
we have the estimate value $y_i^j$ for $x_i^j$ and know beforehand that an 
event in the $j$th group occurs.
Because $x_i^j \le T_j$, it is guaranteed that $s(x_i^j,y_i^j) \ge 0$.

With the quantities, we define the function
\[
 \tilde{S}(X,Y)=\sum_{j=1}^m \sum_{i=1}^{d_j} s(x_i^j, y_i^j).
\]
This can be regarded as a (nonsymmetric) distance between $X$ and its
prediction $Y$.
\end{definition}
\begin{proposition}\label{propmaxS}
The relation
$\tilde{S}(X,Y)\le \log_2(\underset{j}{\rm max}~d_j)$ holds
if $y_i^j \le 2 x_i^j$ for $\forall i,j$.
\end{proposition}
\begin{proof}
We have $s(x_i^j,y_i^j) \le -x_i^j\log_2 (x_i^j/T_j)$ if $y_i^j \le 2 x_i^j$.
Then,
\[\begin{split}
&\tilde{S}(X,Y) \le \sum_{j=1}^m \sum_{i=1}^{d_j} -x_i^j\log_2 (x_i^j/T_j)\\
&= \sum_{j=1}^m T_j \sum_{i=1}^{d_j} -(x_i^j/T_j)\log_2 (x_i^j/T_j)
\le \sum_{j=1}^m T_j \log_2(d_j)\\
&\le \log_2(\underset{j}{\rm max}~d_j).
\end{split}\]
In this transformation, we have used the fact that, for each $j$, one may
regard $(x_i^j/T_j)$ as probabilities $z_i^j$ satisfying
$\sum_{i=1}^{d_j} z_i^j = 1$.
$\Box$\end{proof}
The introduced function $\tilde{S}(X,Y)$ is employed to
quantify a distance between two collections.
One may be curious about the reason why we have introduced the quantity
$s(x_i^j,y_i^j)$ to define this function.
Indeed, it is more common to use the function\\
$r(x_i^j,y_i^j)=-x_i^j\log_2(x_i^j/y_i^j)$ instead of $s(x_i^j,y_i^j)$ to quantify
a discrepancy between $x_i^j$ and $y_i^j$. For example, the relative entropy
employs $r$. However, $r$ can only be used under the condition
$x_i^j \not = 0 \leftrightarrow y_i^j \not = 0$.
In the theory we are constructing in this contribution, this condition
does not hold.

Now we define the new measure of nonclassical correlation
on the basis of Proposition \ref{prop1}.
\begin{definition}\label{defM}
Suppose there are $m$ distinct eigenvalues,
$\eta_1,\ldots,\eta_m$, for a bipartite state $\rho^{\rm AB}$,
i.e., $\rho^{\rm AB} = \sum_{j=1}^m \tilde{\rho}^{\eta_j}$.
Let us write the dimension of the $\eta_j$ eigenspace as $d^{\eta_j}$.
For $\tilde{\rho}^{\eta_j}$, consider the eigenvalues
$\lambda_i^{j, {\rm A}}$ of
the reduced matrix ${\rm Tr}_{\rm B}\tilde{\rho}^{\eta_j}$.
Let $d_j^{\rm A}$ be the rank of
${\rm Tr}_{\rm B}\tilde{\rho}^{\eta_j}$.
Consider the collections
\[
X^{\rm A}(\rho^{\rm AB})=\{
\{\lambda_i^{1, {\rm A}}\}_{i=1}^{d_1^{\rm A}},
\{\lambda_i^{2, {\rm A}}\}_{i=1}^{d_2^{\rm A}},
\cdots,
\{\lambda_i^{m, {\rm A}}\}_{i=1}^{d_m^{\rm A}}\}
\]
and
\[\begin{split}
 Y^{\rm A}(\rho^{\rm AB})=\{&
\{{\rm nim}_{\eta_1}(\lambda_i^{1, {\rm A}})\}_{i=1}^{d_1^{\rm A}},
\{{\rm nim}_{\eta_2}(\lambda_i^{2, {\rm A}})\}_{i=1}^{d_2^{\rm A}},\\
&\cdots,
\{{\rm nim}_{\eta_m}(\lambda_i^{m, {\rm A}})\}_{i=1}^{d_m^{\rm A}}\}
\end{split}\]
as functions of $\rho^{\rm AB}$.
The measure of nonclassical correlation from the view of the subsystem
${\rm A}$ is defined as
\begin{equation}\label{eqMA}\begin{split}
M^{\rm A}(\rho^{\rm AB})
&=\tilde{S}[X^{\rm A}(\rho^{\rm AB}),Y^{\rm A}(\rho^{\rm AB})]\\
&=-\sum_{j=1}^m \sum_{i=1}^{d_j^{\rm A}}\left|
\lambda_i^{j, {\rm A}}-{\rm nim}_{\eta_j}(\lambda_i^{j, {\rm A}})\right|
\log_2 \frac{\lambda_i^{j, {\rm A}}}{\eta_j d^{\eta_j}}.
\end{split}\end{equation}
In the same way, $M^{\rm B}(\rho^{\rm AB})$ is defined:
\[
M^{\rm B}(\rho^{\rm AB})
=\tilde{S}[X^{\rm B}(\rho^{\rm AB}),Y^{\rm B}(\rho^{\rm AB})],
\]
where
\[
X^{\rm B}(\rho^{\rm AB})=\{
\{\lambda_i^{1, {\rm B}}\}_{i=1}^{d_1^{\rm B}},
\{\lambda_i^{2, {\rm B}}\}_{i=1}^{d_2^{\rm B}},
\cdots,
\{\lambda_i^{m, {\rm B}}\}_{i=1}^{d_m^{\rm B}}\}
\]
and
\[\begin{split}
 Y^{\rm B}(\rho^{\rm AB})=\{&
\{{\rm nim}_{\eta_1}(\lambda_i^{1, {\rm B}})\}_{i=1}^{d_1^{\rm B}},
\{{\rm nim}_{\eta_2}(\lambda_i^{2, {\rm B}})\}_{i=1}^{d_2^{\rm B}},\\
&\cdots,
\{{\rm nim}_{\eta_m}(\lambda_i^{m, {\rm B}})\}_{i=1}^{d_m^{\rm B}}\}
\end{split}\]
with $\lambda_i^{j, {\rm B}}$ the eigenvalues of
the reduced matrix ${\rm Tr}_{\rm A}\tilde{\rho}^{\eta_j}$ and
$d_j^{\rm B}$ the rank of
${\rm Tr}_{\rm A}\tilde{\rho}^{\eta_j}$.

The new measure of nonclassical correlation is defined by their average
\[
M(\rho^{\rm AB})=\frac{1}{2}
[M^{\rm A}(\rho^{\rm AB})+M^{\rm B}(\rho^{\rm AB})].
\]
\end{definition}
Let us list the basic properties of this measure.
\begin{proposition}\label{propvgen}
The measure $M$ vanishes for any (bipartite) state having a product eigenbasis.
\end{proposition}
\begin{proof}
By Proposition \ref{prop1}, for any state having a product eigenbasis,
$\lambda_i^{j, {\rm A}} =
{\rm nim}_{\eta_j}(\lambda_i^{j, {\rm A}})$ and
$\lambda_i^{j, {\rm B}} =
{\rm nim}_{\eta_j}(\lambda_i^{j, {\rm B}})$
hold. This proves the proposition.
$\Box$\end{proof}
\begin{proposition}\label{propv2}
For a state $\rho^{AA'}$ of an $N\times N$-dimensional bipartite
system, $M(\rho^{\rm AA'})\le \log_2 N$.
\end{proposition}
\begin{proof}
Under this condition, we have $\underset{j}{\rm max}~d_j^{\rm A(B)}\le N$.
In addition, generally ${\rm nim}_y(x)\le 2x$ for $0\le x,y$.
Thus by Proposition \ref{propmaxS}, this proposition holds.
$\Box$\end{proof}
There are some properties of the measure for the pure states.
\begin{proposition}\label{propvpure}
(i) The measure $M$ vanishes for a pure state if and only if the state is a
product state. Thus $M$ never vanishes for an entangled pure state.\\
(ii) For a pure state $|\phi\rangle\langle\phi|$, $M$ is bounded
above by the entropy of entanglement,
$S_{\rm vN}({\rm Tr_B}|\phi\rangle\langle\phi|)$, where
$S_{\rm vN}(\sigma)=-\sum_{k=1}^d\lambda_k\log_2\lambda_k$ is
the von Neumann entropy for a general $d$-dimensional density matrix
$\sigma$ with the eigenvalues $\lambda_k$.\\
(iii) The measure $M$ reduces to the entropy of entanglement
for a pure state with the Schmidt coefficients $\le 1/\sqrt{2}$.
\end{proposition}
\begin{proof}
(i) Consider the Schmidt decomposition of a pure state vector
$|\phi\rangle$, namely, $|\phi\rangle=\sum_{k=1}^{R}
\sqrt{c_k}|a_k\rangle|b_k\rangle$
with $\sqrt{c_k}$'s the Schmidt coefficients, $R$ the Schmidt rank, and
$|a_k\rangle$'s and $|b_k\rangle$'s being the eigenvectors of the
reduced density matrix ${\rm Tr_B}|\phi\rangle\langle\phi|$ and those of 
${\rm Tr_A}|\phi\rangle\langle\phi|$, respectively,
corresponding to the nonzero eigenvalues $c_k$ ($c_k$'s are common for
the reduced density matrices). We have only to consider the value $1$
for $\eta$ in the calculation of $M$ for a pure state. Then,
obviously, $M$ vanishes if and only if the Schmidt rank is $1$.\\
(ii) Consider the real numbers $0\le x,y\le 1$. The relation
${\rm nim}_y(x) \le 2x$ holds because $0$ is closer or equally close to
 $x$ than $2x$. Therefore,
$-|x-{\rm nim}_y(x)|\log_2 x \le -x\log_2 x$. In addition, for a general
pure state $|\phi\rangle\langle\phi|$, we have only to consider the
value $1$ for $\eta$. Therefore,
$M(|\phi\rangle\langle\phi|)=M^{\rm A}(|\phi\rangle\langle\phi|)
=M^{\rm B}(|\phi\rangle\langle\phi|)
\le S_{\rm vN}({\rm Tr_B}|\phi\rangle\langle\phi|)$.\\
(iii) Recall again that, for a pure state, we have only to
consider the value $1$ for $\eta$.
In case we have a pure state vector $|\phi'\rangle$ whose Schmidt coefficients
are $\le 1/\sqrt{2}$, the definition obviously
reduces to
$M(|\phi'\rangle\langle\phi'|)=M^{\rm A}(|\phi'\rangle\langle\phi'|)
=M^{\rm B}(|\phi'\rangle\langle\phi'|)
=S_{\rm vN}({\rm Tr_B}|\phi'\rangle\langle\phi'|)$.
$\Box$\end{proof}
For the typical two-qubit entangled state $|\phi_p\rangle =
\sqrt{p}|00\rangle+\sqrt{1-p}|11\rangle$, $M(|\phi_p\rangle\langle\phi_p|)$
behaves as illustrated in Fig.\ \ref{figphip}.
\begin{figure}
\begin{center}
\scalebox{0.6}{\includegraphics{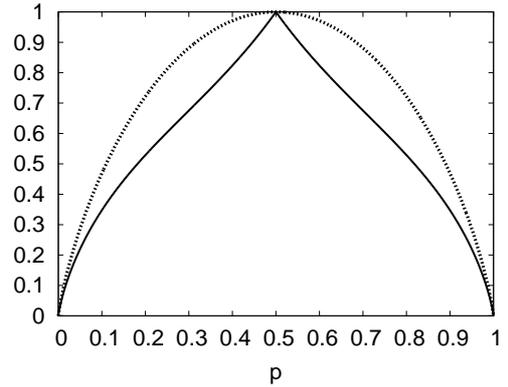}}
\caption{\label{figphip}Plots of $M(|\phi_p\rangle\langle\phi_p|)$
[the lower (solid) curve] and the entropy of entanglement $S_{\rm vN}({\rm Tr_B}
|\phi_p\rangle\langle\phi_p|)$ [the upper (dotted) curve] as functions
of $p$.}
\end{center}
\end{figure}
It should be noted that, for general states, $M$ is not an entanglement
measure but a measure of nonclassical correlation. Because of the
above properties, $M$ can be seen as an entanglement measure in case
the states are limited to pure states.

As important properties, for general states, the measure $M$ satisfies
local-unitary invariance and it takes its maximum value for the
generalized Bell states as we will see in Section~\ref{secinvar}.

The improvement in complexity is significant as we have already mentioned:
the measure is computed by using the eigenvalues of at
most $2d^{\rm A}d^{\rm B}$ reduced matrices. The total
complexity is dominated by the complexity of diagonalizing the
original density matrix, $O({d^{\rm A}}^3 {d^{\rm B}}^3)$,
which is larger than the complexity of tracing out a subsystem for
each truncated density matrix.
In the exceptional cases where $d^{\rm A}>{d^{\rm B}}^2$
or $d^{\rm B}>{d^{\rm A}}^2$, the complexity of diagonalizing all the
reduced matrices,\\
$O[{\rm max}({d^{\rm A}}^4 d^{\rm B}, d^{\rm A}{d^{\rm B}}^4)]$,
becomes the largest cost.
Thus the total complexity, namely, the number of the floating-point
operations in total to compute the measure, is
\[
O[{\rm max}({d^{\rm A}}^3 {d^{\rm B}}^3,~
  {d^{\rm A}}^4 d^{\rm B}, ~d^{\rm A}{d^{\rm B}}^4)].
\]

In the following examples, one may recognize the simpleness of
the process to compute $M$.
\paragraph{Examples}
Consider the two-qubit state
\begin{equation}\label{varsigma}
\varsigma=\frac{1}{2}(|00\rangle\langle00|+|1+\rangle\langle1+|)
\end{equation}
with $|+\rangle=(|0\rangle+|1\rangle)/\sqrt{2}$. This has no product
eigenbasis while it is separable.
The nonzero eigenvalue of $\varsigma$ is $1/2$ with the multiplicity two.
The eigenvalue of ${\rm Tr_B}\varsigma$ is $1/2$ with the multiplicity two;
this leads to $M^{\rm A}(\varsigma)=\tilde{S}(\{\{1/2,1/2\}\},\{\{1/2,1/2\}\}) = 0$.
The eigenvalues of ${\rm Tr_A}\varsigma$ are $(2\pm\sqrt{2})/4$
for which the nearest integer multiples of $1/2$ are $0$ and $1$;
this leads to $M^{\rm
B}(\varsigma)=\tilde{S}(\{\{(2-\sqrt{2})/4,(2+\sqrt{2})/4\}\},\{\{0,1\}\})\simeq 0.439$.
Therefore $M(\varsigma)=M^{\rm B}(\varsigma)/2 \simeq 0.220$.

In view of the problem of detecting nonclassical correlation, it is
easy for this example since one of the reduced density matrices is
nondegenerate (see the discussion in Section~\ref{sec4}). The next
example is not an easy case in this sense.

Consider the state in a $4\times 4$ dimensional system,
\[
 \zeta = \frac{1}{4}\left(
  |00 \rangle\langle 00|
+ |+2 \rangle\langle +2| + |2+ \rangle\langle 2+|
+ |33 \rangle\langle 33|
\right).
\]
Let $U_\zeta=U^{\rm A}_\zeta\otimes U^{\rm B}_\zeta$ be some local unitary
operation. Suppose that $\zeta'= U_\zeta \zeta U_\zeta^\dagger$ is
given and $M$ is computed for it.
The global and local eigenvalues are unchanged but the matrix form may
be complicated at a glance due to $U_\zeta$. We have the same eigenvalues
for ${\rm Tr_B}\zeta'$ and ${\rm Tr_A}\zeta'$. The eigenvalues of
${\rm Tr_B}\zeta'$ are $1/4$ with multiplicity two and
$(2\pm\sqrt{2})/8$. This leads to
$M(\zeta')=M^{\rm A}(\zeta')=M^{\rm B}(\zeta')=
\tilde{S}(\{\{(2-\sqrt{2})/8,(2+\sqrt{2})/8,1/4,1/4\}\},\\\{\{0,1/2,1/4,1/4\}\})
\simeq 0.366$.

The examples we have seen are fixed states.
In case there are parameters, it is rather complicated to write
$M$ in terms of them because it is influenced by the
degeneracy of eigenvalues and it involves the function ${\rm nim}_y(x)$.
As $M$ is tractable, one seldom has a motivation to decompose $M$ analytically.
One should keep $M$ as a routine in a computational program, which
might be a way to avoid the complication. This is in contrast to
intractable measures, for which finding an analytical solution for a
particular form of density matrices is highly motivated. As for quantum
discord, which is in general intractable to compute, analytical
solutions are known for particular forms of density matrices.
Luo \cite{L08-2} gave a general solution for the two-qubit density matrices
in the form $\kappa=(I\otimes I+\sum_{j=x,y,z}c_j\sigma_j\otimes\sigma_j)/4$,
where $\sigma_j$'s are Pauli matrices and $c_j$'s are real parameters. Ali {\em et al.} \cite{ARA10}
very recently gave a general solution for the two-qubit denstiy matrices
with only diagonal and anti-diagonal elements.

The density matrix $\kappa$ is special in the sense that its
eigenvectors are the Bell basis vectors. With this fact and its eigenvalues
$(1-c_x-c_y-c_z)/4$, $(1-c_x+c_y+c_z)/4$, $(1+c_x-c_y+c_z)/4$,
and $(1+c_x+c_y-c_z)/4$, it is straight-forward albeit complicated to
write $M(\kappa)$ in terms of the parameters $c_j$.

\subsection{Imperfectness in the detection range}
We have achieved a reduction in the complexity by introducing the measure $M$.
We have shown that several properties are satisfied by $M$ in Propositions
\ref{propvgen}, \ref{propv2}, and \ref{propvpure}.
As a desirable property for a measure of nonclassical correlation,
it never vanishes for entangled pure states.
However, our main concern is to use $M$ for general states which are
mostly mixed states. As a matter of fact, the measure does not have a
perfect detection range as is expected from the fact that it does not
test all the local bases unlike other expensive measures like quantum
discord (with both-side test) or zero-way quantum
deficit. For example, the measure $M$ vanishes for the two-qubit state
\begin{equation}\label{sigma}\sigma=
\frac{1}{6}\left(|00\rangle\langle00|+2|01\rangle\langle01|
+3|1+\rangle\langle1+|\right).
\end{equation}
Nevertheless, this state has no product eigenbasis because
$|0\rangle\langle 0|$ and $|+\rangle\langle +|$ cannot be
diagonalized in the same basis.

Consequently, what one can claim is that a state for which the measure does
not vanish is in the outside of the set B of the states for which the
measure $M$ vanishes, and hence in the outside of the set C
of the states having a product eigenbasis, as illustrated in
Fig.~\ref{classfig}.
Note that the set B includes some inseparable states while C does not.
For example, the measure $M$ vanishes for the state $\tau$, which is represented
as a density matrix acting on the $(3\times 3)$-dimensional Hilbert space:
\begin{equation}\label{tau}
 \tau=\frac{1}{3}(|\phi\rangle^{\rm AB}\langle\phi|
+|\psi\rangle^{\rm AB}\langle\psi|
+|\xi\rangle^{\rm AB}\langle\xi|)
\end{equation}
with
\begin{eqnarray*}
|\phi\rangle^{\rm AB}&=\frac{|0\rangle^{\rm A}|1\rangle^{\rm B}
+|1\rangle^{\rm A}|0\rangle^{\rm B}}{\sqrt{2}},\\
|\psi\rangle^{\rm AB}&=\frac{|1\rangle^{\rm A}|2\rangle^{\rm B}
+|2\rangle^{\rm A}|1\rangle^{\rm B}}{\sqrt{2}},\\
|\xi\rangle^{\rm AB}&=\frac{|2\rangle^{\rm A}|0\rangle^{\rm B}
+|0\rangle^{\rm A}|2\rangle^{\rm B}}{\sqrt{2}}.
\end{eqnarray*}
This is because the nonzero eigenvalue of $\tau$ is $1/3$ with
the multiplicity three and the eigenvalue of ${\rm Tr_B}\tau=
{\rm Tr_A}\tau$ is also $1/3$ with the multiplicity three.
The state $\tau$ is inseparable because its partial transposition
$(I\otimes T)\tau$  (here, $T$ is the transposition map) has the
eigenvalues $-1/6$, $1/6$, and $1/3$ with multiplicities
two, six, and one, respectively. Thus it has been found that $M$
cannot detect nonclassical
correlation of this negative-partial-transpose (NPT) state.
\begin{figure}[tbph]
\begin{center}
\scalebox{0.8}{\includegraphics{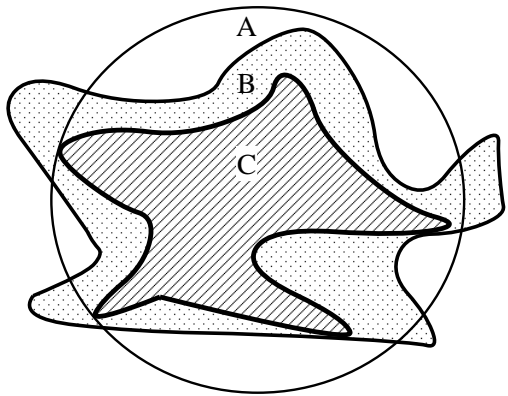}}
\caption{\label{classfig}
The hierarchy of the quantum states. A: The convex set of the separable
states. B: The nonconvex set of the states for which the measure $M$
vanishes. C: the nonconvex set of the states having a product
eigenbasis (${\rm C}\subset {\rm A}\cap {\rm B}$). Internal geometric
structures of each set are not depicted, which should wait for future
investigations.
}
\end{center}
\end{figure}
This example however does not weaken the measure $M$ very much in light
of the benefit of computational tractability. 
In addition, regarding $M$ as a detection tool, there is a way to
extend the detection range simply, which we will discuss later in
Section~\ref{sec4}.

\subsection{Relative detection ability}
One can compare the measures $M$ and $G$ in their detection ability using
the state $\sigma$ given in Eq.~(\ref{sigma}) and another particular state
for a couple of qubits.

The measure $M$ vanishes for $\sigma$ while $G$ does not vanish.
We have $G(\sigma)=H(1/3)-H[(6-\sqrt{10})/12]\simeq 0.129$ where
$H(x)=-x\log_2 x - (1-x)\log_2(1-x)$ is the binary entropy function.

On the other hand, $M$ does not vanish for the state
$\sigma'=|\phi\rangle\langle\phi|/2
+(|01\rangle\langle01|+|10\rangle\langle10|)/4$
with $|\phi\rangle=(|00\rangle+|11\rangle)/\sqrt{2}$, while $G$
vanishes. We have
$M(\sigma') = M^{\rm A}(\sigma')=M^{\rm B}(\sigma')=\\
\tilde{S}(\{\{1/4,1/4\},\{1/4,1/4\}\},\{\{0,0\},\{1/4,1/4\}\})=1/2$.

Therefore, $M$ is neither stronger nor weaker than $G$ in detecting
nonclassical correlations.

\subsection{Axioms satisfied by the measure}\label{secinvar}
As we have described in Section~\ref{sec1}, there are axioms that
are desired to be satisfied by a measure of nonclassical correlation.
Here, we examine the measure $M$ in this regard.

First, it is easy to verify that $M$ vanishes for any state having a
product eigenbasis as we have already shown in Proposition\ \ref{propvgen}.

Second, we show that $M$ is invariant under the local unitary operations
$\mathcal{U}^{\rm A}\otimes \mathcal{U}^{\rm B}:
\rho\mapsto (U^{\rm A}\otimes U^{\rm B})\rho
(U^{\rm A}\otimes U^{\rm B})^\dagger$.
This is easily verified from the fact that (i)
${\rm Tr_{B}}\tilde{\rho}^{\eta}$
and
${\rm Tr_{B}}(\mathcal{U}^{\rm A}\otimes \mathcal{U}^{\rm B})
(\tilde{\rho}^{\eta})$ have the same eigenvalues, and similarly,
(ii) ${\rm Tr_{A}}\tilde{\rho}^{\eta}$
and
${\rm Tr_{A}}(\mathcal{U}^{\rm A}\otimes \mathcal{U}^{\rm B})
(\tilde{\rho}^{\eta})$ have the same eigenvalues.

Third, it is desirable that a measure takes its maximum value for certain
maximally entangled states as we have mentioned in Section~\ref{sec1}.
The dimensions $d^{\rm A}$ and $d^{\rm B}$ of subsystems ${\rm A}$
and ${\rm B}$ are set to $N$ to consider the maximum value. 
Then, by Proposition~\ref{propv2}, $M(\rho^{\rm AB})\le \log_2 N$ holds.
It is now easy to show that $M$ takes its maximum value for the
generalized Bell state
$|\Psi_{\rm Bell}\rangle\langle\Psi_{\rm Bell}|$
with
\[
|\Psi_{\rm Bell}\rangle =
\frac{1}{\sqrt{N}}\sum_{i=1}^{N}|i\rangle^A |i\rangle^B.
\]
This is a straightforward calculation.

It has been shown that the basic axioms are satisfied. Let us now
mention that, in addition to these axioms, some mathematical properties,
such as convexity and additivity properties are often investigated
for a measure. For a measure in the present context, convexity
should not hold because the set of the classically correlated
states is a nonconvex set. The additivity properties are next
investigated.

\subsection{Investigation on additivity properties}
The measure $M$ has neither the subadditivity property nor
the superadditivity property as we prove below.
It is also shown not to be weakly additive.

Here, let us denote a splitting of a system considered for the measure $M$
by a subscript $\cdot|\cdot$.
\begin{proposition}
Neither the subadditivity
\[
M_{\rm AC|BD}(\rho^{\rm AB}\otimes \sigma^{\rm CD})
\le M_{\rm A|B}(\rho^{\rm AB})+ M_{\rm C|D}(\sigma^{\rm CD})
\]
nor the superadditivity
\[
M_{\rm AC|BD}(\rho^{\rm AB}\otimes \sigma^{\rm CD})
\ge M_{\rm A|B}(\rho^{\rm AB})+ M_{\rm C|D}(\sigma^{\rm CD})
\]
holds. In addition, the weak additivity
\[
 M_{\rm AA\cdots|BB\cdots}
({\rho^{{\rm A}{\rm B}}}^{\otimes m})
= m M_{\rm A|B}(\rho^{\rm AB})
\]
does not hold.
\end{proposition}
\begin{proof}
(i) First we prove that subadditivity does not hold.
Consider the state $\xi=\sigma^{\rm AB}\otimes \sigma^{\rm CD}$
with the state $\sigma$ defined by Eq.~(\ref{sigma}).
We have already found $M_{\rm A|B}(\sigma^{\rm AB})=0$.
Now we calculate $M_{\rm AC|BD}(\xi)$. The state
$\xi$ has the eigenvalues $e=0,1/36,1/18,1/12,1/9,1/6$, and $1/4$.
Let us write the truncated density matrix down to the $e$-eigenspace
as ${\tilde \xi}^{e}$.
We have ${\tilde \xi}^{1/12}=
(|001+\rangle\langle001+|+|1+00\rangle\langle1+00|)/12$.
This leads to
${\rm Tr_{AC}}{\tilde \xi}^{1/12}=
(|0+\rangle\langle0+|+|+0\rangle\langle+0|)/12$ whose eigenvalues are
$1/24, 1/8$, and $0$ with multiplicity two.
Similarly, ${\tilde \xi}^{1/6}=
(|011+\rangle\langle011+|+|1+01\rangle\langle1+01|)/6$ and
${\rm Tr_{AC}}{\tilde \xi}^{1/6}$ has the eigenvalues $1/12,1/4$, and
$0$ with multiplicity two.
For other ${\tilde \xi}^{e}$'s, ${\rm Tr_{AC}}{\tilde \xi}^{e}$
has the eigenvalues equal to $e$.
Therefore, $M_{\rm AC|BD}^{\rm BD}(\xi)=\\
\tilde{S}(\{\{1/24,1/8\},\{1/12,1/4\}\},\{\{0, 1/12\},\{0,1/6\}\})
\simeq 0.302$.
In addition, it is easy to find $M_{\rm AC|BD}^{\rm AC}(\xi)=0$
because ${\rm Tr_{BD}}{\tilde \xi}^{e}$ has the eigenvalues
equal to integer multiples of $e$ for every $e$.
Consequently, $M_{\rm AC|BD}(\xi)=0.151$, which is larger than
 $M_{\rm A|B}(\sigma^{\rm AB})+M_{\rm C|D}(\sigma^{\rm CD})=0$.
This is a counterexample to subadditivity.\\
(ii) Second we prove that superadditivity does not hold.
Consider the state $\xi'={\sigma''}^{\rm AB}\otimes {\sigma''}^{\rm CD}$
with ${\sigma''}=(1/4)|\phi\rangle\langle\phi|
+(3/8)(|01\rangle\langle01|+|10\rangle\langle10|)$ where
$|\phi\rangle=(|00\rangle+|11\rangle)/\sqrt{2}$.
Let us first calculate $M_{\rm A|B}({\sigma''}^{\rm AB})$. We have
${\tilde {\sigma''}}^{1/4}=(1/4)|\phi\rangle\langle\phi|$
and
${\tilde {\sigma''}}^{3/8}=(3/8)(|01\rangle\langle01|+|10\rangle\langle10|)$.
The eigenvalues of ${\rm Tr_B}{\tilde {\sigma''}}^{1/4}$
are $1/8$ with multiplicity two and those of
${\rm Tr_B}{\tilde {\sigma''}}^{3/8}$ are $3/8$ with multiplicity two.
Because of the symmetry of the state, this leads to
$M_{\rm A|B}({\sigma''}^{\rm AB})=M_{\rm A|B}^{\rm A}({\sigma''}^{\rm AB})
=M_{\rm A|B}^{\rm B}({\sigma''}^{\rm AB})=-2\times (1/8)\log_2[(1/8)/(1/4)]
=1/4$. Let us second calculate $M_{\rm AC|BD}(\xi')$.
We have ${\tilde {\xi'}}^{1/16}=(1/16)(|\phi\rangle\langle\phi|)^{\otimes 2}$,
${\tilde {\xi'}}^{3/32}=(3/32)[|\phi\rangle\langle\phi|\otimes
(|01\rangle\langle01|+|10\rangle\langle10|)+
(|01\rangle\langle01|+|10\rangle\langle10|)\otimes|\phi\rangle\langle\phi|]$,
and ${\tilde {\xi'}}^{9/64}=
(9/64)(|01\rangle\langle01|+|10\rangle\langle10|)^{\otimes 2}$.
The eigenvalues of ${\rm Tr_{BD}}{\tilde {\xi'}}^e$ are
$1/64$ for $e=1/16$, $3/16$ for $e=3/32$, and $9/64$ for $e=9/64$,
with multiplicity four for each $e$. Because of the symmetry of the
state, this leads to
$M_{\rm AC|BD}(\xi')=M_{\rm AC|BD}^{\rm AC}(\xi')
=M_{\rm AC|BD}^{\rm BD}(\xi')=-4\times (1/64)\log_2[(1/64)/(1/16)]
=1/8$, which is less than
$M_{\rm A|B}({\sigma''}^{\rm AB})+M_{\rm C|D}({\sigma''}^{\rm CD})=1/2$.
This is a counterexample to superadditivity.\\
(iii) The above counterexamples shown in (i) and (ii) are
also counterexamples to weak additivity.
$\Box$\end{proof}

\section{Discussions}\label{sec4}
The main aim of introducing the measure $M$ has been the computational
tractability. The problem of quantification is harder than detection as
certain axioms should be satisfied as we discuss later in this section.
In view of a detection problem rather than quantification, indeed, it is
known to be possible to decide whether a given density matrix has a
product eigenbasis within polynomial time for some special cases:
(i) In case there are only nondegenerate eigenvalues for a given
density matrix, it has a product eigenbasis if and only if the
Schmidt ranks of the eigenvectors $|v\rangle^{\rm AB}$ are all one, i.e.,
$|v\rangle^{\rm AB}=|a\rangle^{\rm A}|b\rangle^{\rm B}$ and, for each
subsystem, the Schmidt vectors ($|a\rangle$'s for A and $|b\rangle$'s
for B) are either orthogonal or equal to each other, neglecting the global
phase difference.
(ii) In case the reduced density matrices of a given density matrix have no
degeneracy, the local eigenbases are uniquely determined neglecting the global
phase factors. Then, the density matrix has a product eigenbasis if and only
if the product of the local eigenbases is the eigenbasis of the total system.
(iii) The case only one of the reduced density matrices, say, one for
the subsystem ${\rm B}$, is nondegenerate is also an easy case. In this
case, for the reduced density matrix, we have the eigenvectors
$|v_j\rangle^{\rm B}$ that are unique neglecting their global phase factors.
Then, first we test if the density matrix $\rho^{\rm AB}$ is equal
to $\sum_j {}^{\rm B}\langle v_j|\rho^{\rm AB}|v_j\rangle^{\rm B}
\otimes |v_j\rangle^{\rm B}\langle v_j|$. If this is false, then
$\rho^{\rm AB}$ has no product eigenbasis. If true, then
$\rho^{\rm AB}$ has a product eigenbasis if and only if
${}^{\rm B}\langle v_j|\rho^{\rm AB}|v_j\rangle^{\rm B}$'s are commutative
to each other.
(iv) There is a property very recently mentioned \cite{F10}:
the commutation relations $[\rho^{\rm AB},{\rm Tr_{B}}\rho^{\rm
AB}\otimes I^{\rm B}] = [\rho^{\rm AB},
I^{\rm A}\otimes{\rm Tr_{A}}\rho^{\rm AB}]= 0$ hold
if $\rho^{\rm AB}$ has a product eigenbasis. Any state that does not
satisfy these relations has no product eigenbasis while the converse does
not hold in general. For example, this detection method does not work
for the generalized Bell states.

Regarding measures as detection tools, combinations of imperfect
measures, that are neither stronger nor weaker to each other, makes
a tool with a larger detection range. One may utilize the measures
$M$, $G$, and some entanglement measure to produce a detection tool
easily, which is nonvanishing for entangled states detectable by the
entanglement measure.

Some existing measures, namely, the measurement-induced disturbance
\cite{G07,L08} and its variants, are computable within polynomial
time in case the dephasing basis is uniquely determined, namely in the
case (ii), while they are not otherwise. The measure $M$ is, in contrast,
always computable within polynomial time.

As is expected for a tractable measure, $M$ is imperfect in its
detection range. The measure has been found to vanish not only for the
state (\ref{sigma}) but for the NPT state (\ref{tau}). It has been
controversial as to which extent computational tractability has the
higher priority, facing the trade-off with the detection range. As for
entanglement measures, negativity and the logarithmic
negativity \cite{neg:i} are commonly used although they are not perfect
in the detection range. The basic axiom of monotonicity has been a
ground of argument for entanglement measures. Thus there must be a
demand of measures of nonclassical correlation whose detection ranges
may be imperfect as long as certain axioms are satisfied. As we have
described in Section \ref{sec1}, there are possible axioms for nonclassical
correlation measures:
(i) a measure should vanish for any classically correlated state,
(ii) a measure should be invariant under local unitary transformations,
and (iii) a measure should have the maximum value for certain maximally
entangled states. These are all satisfied by $M$ as described in
Section \ref{secinvar}.

For a more quantitative discussion, one may prefer to have the volume of
the detection range not very smaller than the volume of the set of the
nonclassically correlated states. This has not been studied and it is
an open problem in the present stage. We need to start from
understanding the geometric structure of the sets of the
classically/nonclassically correlated states. Quite recently, a very
simplified discussion on the geometric structure is reported in
Ref.\ \cite{F10}.

As an additional topic, there will be a particular behavior when an
evolution process is evaluated by $M$. The value of $M$ has an abrupt
change when multiple eigenvalues of the density matrix gradually change
and cross with each other. At the crossing point, the value of $M$ most
probably has discontinuity. This is because its component values are
computed inside each eigenspace. Thus $M$ has an abrupt change when
multiple eigenspaces are admixed in a process. This behavior corresponds
to a physical event of eigenvalue crossing, but can be seen as an
unstable behavior. It is controversial if this is counted as a drawback
as a measure.

Let us turn into a rather conceptual problem. It is often of general
interest to find a physical interpretation of a measure to justify
its quantification. Here we suggest an interpretation of the definition
of the measure $M$. It involves a transmission of a system from
a dealer to players and guess works by the players on local information.

Consider a thermal state $\rho_{\rm th} = e^{-\beta H}$ with a system
Hamiltonian $H$ and the reciprocal temperature $\beta$. It is possible
to have a transmission line (namely, a waveguide or a resonator) in
resonance with an eigenfrequency $\nu$ of $H$ under the same
temperature. An eigenvalue $\eta_\nu$ of $\rho_{\rm th}$ can be written
as $\eta_\nu = e^{-\beta \nu}$. This transmission line acts as a channel
that projects a density matrix of the system onto an eigenspace of
$H$ corresponding to the eigenfrequency $\nu$. Let us mention that one
need not to have a transmission line directly acting on the system in
case this is difficult for some physical setup. Consider the case where
the system is a molecular spin system placed in a magnetic resonance
spectrometer, for example. One may attach band-pass filters to its probe
circuits so that only the frequency band around the signal frequency
corresponding to $\nu$ is captured. This virtually behaves similarly as the
channel.

Suppose a dealer has a bipartite system ${\rm AB}$ in the thermal state
$\rho^{\rm AB}=\rho_{\rm th}$ and sends it through the channel. The
surviving state in the channel is $\tilde{\rho}^{\eta_\nu}$ (in the
notation used in Definition~\ref{def1}). At the other side of the channel,
Alice takes the subsystem ${\rm A}$ and Bob takes the subsystem ${\rm B}$.
These players try to guess their local eigenvalues of the reduced
matrices ${\rm Tr_B}\tilde{\rho}^{\eta_\nu}$ and
${\rm Tr_A}\tilde{\rho}^{\eta_\nu}$,  respectively.
A dealer may answer to the query as to whether or not a value is equal
to a local eigenvalue of the subsystem indicated in the query. Unlike the
scenario for the measure $G$ described in
Section~\ref{sec2}, suppose players want to use a linear-time strategy.
A natural strategy is to guess a local eigenvalue as an integer multiple
of $\eta_\nu$. They can easily find the value of
$\eta_\nu = e^{-\beta \nu}$ from the resonance frequency $\nu$ and the
temperature of the transmission line. Because a local eigenvalue does not
exceed $\eta_\nu d^{\eta_\nu}$, the maximum number of queries they
can try is $1+d^{\eta_\nu}$ for one eigenvalue. The quantity
$M^{\rm A (B)}(\rho^{\rm AB})$ consists of the components
\[
 -\sum_{i=1}^{d_\nu^{\rm A(B)}}\left|
\lambda_i^{\nu, {\rm A(B)}}-{\rm nim}_{\eta_\nu}(\lambda_i^{\nu, {\rm A(B)}})
\right| \log_2 \frac{\lambda_i^{\nu, {\rm A(B)}}}{\eta_\nu d^{\eta_\nu}}
\]
found in Eq.\ (\ref{eqMA}) (we read $j$ therein as $\nu$ here).
Each component is a discrepancy between the set of the true
local eigenvalues and that of their nearest guessed values for Alice
(Bob) in this strategy for $\nu$ with the reduction factor
$\log_2(\eta_\nu d^{\eta_\nu}) < 0$
(recall that ${\rm Tr}\tilde{\rho}^{\eta_\nu} =\eta_\nu d^{\eta_\nu}$)
corresponding to players' knowledge on which channel is used.

In this way, the definition of $M$ has been interpreted in view of a
physical process. The interpretation is, however, not applicable for general
states other than the thermal state. Although the system Hamiltonian is
$H$, the quantum state $\rho^{\rm AB}$ can be changed from the thermal
state by a unitary operation. Then the correspondence between the
spectrum of $\rho^{\rm AB}$ and that of $H$ is broken.
It is to be hoped that a different protocol-based interpretation will
be found for $M$ for general states.
For the time being, the validity of $M$ as a measure for general states 
relies on those axioms that we have discussed.

\section{Summary}\label{sec5}
We have proposed an unconventional measure of nonclassical correlation
by using truncations of a density matrix, on the basis of
Proposition \ref{prop1}. The mathematical properties of the measure
have been investigated.
It is invariant under local unitary operations and it takes
the maximum value for the generalized Bell states while it is imperfect in
the detection range and it has no additivity property.
It is usable for a practical evaluation of quantum states
because it is calculated within polynomial time in the dimension of a
density matrix.

\begin{acknowledgements}
~~~~The authors are thankful to Karol\\ {\.{Z}}yczkowski for helpful comments
on a geometric structure related to Fig.~\ref{classfig}.
They are also thankful to Todd Brun for a comment on the
definition of the measure $M$.
A.S. and M.N. are supported by the ``Open Research Center'' Project
for Private Universities: matching fund subsidy from MEXT.
R.R. is supported by the FIRST program of JSPS.
A.S. is supported by the Grant-in-Aid for Scientific
Research from JSPS (Grant No. 21800077).
R.R. and M.N. have been supported by the Grants-in-Aid for Scientific
Research from JSPS (Grant Nos. 1907329 and 19540422, respectively).
\end{acknowledgements}
\bibliographystyle{spmpsci-mod}
\bibliography{refs_nonclassical}

\begin{thebibliography}{10}
\providecommand{\url}[1]{{#1}}
\providecommand{\urlprefix}{URL }
\expandafter\ifx\csname urlstyle\endcsname\relax
  \providecommand{\doi}[1]{DOI~\discretionary{}{}{}#1}\else
  \providecommand{\doi}{DOI~\discretionary{}{}{}\begingroup
  \urlstyle{rm}\Url}\fi

\bibitem{ARA10}
Ali, M., Rau, A.R.P., Alber, G.: Quantum discord for two-qubit x states.
\newblock Phys.\ Rev.\ A \textbf{81}, 042105 (2010)

\bibitem{BZ06}
Bengtsson, I., \.{Z}yczkowski, K.: Geometry of Quantum States: An Introduction
  to Quantum Entanglement.
\newblock Cambridge University Press, Cambridge (2006)

\bibitem{B99}
Bennett, C.H., DiVincenzo, D.P., Fuchs, C.A., Mor, T., Rains, E., Shor, P.W.,
  Smolin, J.A., Wootters, W.K.: Quantum nonlocality without entanglement.
\newblock Phys.\ Rev.\ A \textbf{59}, 1070--1091 (1999)

\bibitem{BC10}
Bylicka, B., Chru\'sci\'nski, D.: Witnessing quantum discord in {$2\times {N}$}
  systems.
\newblock Phys.\ Rev.\ A \textbf{81}, 062102 (2010)

\bibitem{D09}
Datta, A., Gharibian, S.: Signatures of nonclassicality in mixed-state quantum
  computation.
\newblock Phys.\ Rev.\ A \textbf{79}, 042325 (2009)

\bibitem{D78}
Davies, E.: Information and quantum measurement.
\newblock IEEE\ Trans.\ Inf.\ Theory \textbf{24}, 596--599 (1978)

\bibitem{F10}
Ferraro, A., Aolita, L., Cavalcanti, D., Cucchietti, F.M., Acin, A.: Almost all
  quantum states have nonclassical correlations.
\newblock Phys.\ Rev.\ A \textbf{81}, 052318 (2010)

\bibitem{G07}
Groisman, B., Kenigsberg, D., Mor, T.: "quantumness" versus "classicality" of
  quantum states.
\newblock {\rm arXiv:quant-ph/0703103}  (2007)

\bibitem{GPW05}
Groisman, B., Popescu, S., Winter, A.: Quantum, classical, and total amount of
  correlations in a quantum state.
\newblock Phys.\ Rev.\ A \textbf{72}, 032317 (2005)

\bibitem{HV01}
Henderson, L., Vedral, V.: Classical, quantum and total correlations.
\newblock J.\ Phys.\ A:\ Math.\ Gen. \textbf{34}, 6899--6905 (2001)

\bibitem{H96}
Horodecki, M., Horodecki, P., Horodecki, R.: Separability of mixed states:
  necessary and sufficient conditions.
\newblock Phys.\ Lett.\ A \textbf{223}, 1--8 (1996)

\bibitem{H05}
Horodecki, M., Horodecki, P., Horodecki, R., Oppenheim, J., Sen(De), A., Sen,
  U., Synak-Radtke, B.: Local versus nonlocal information in
  quantum-information theory: Formalism and phenomena.
\newblock Phys.\ Rev.\ A \textbf{71}, 062307 (2005)

\bibitem{LL08}
Li, N., Luo, S.: Classical states versus separable states.
\newblock Phys.\ Rev.\ A \textbf{78}, 024303 (2008)

\bibitem{L08-2}
Luo, S.: Quantum discord for two-qubit systems.
\newblock Phys.\ Rev.\ A \textbf{77}, 042303 (2008)

\bibitem{L08}
Luo, S.: Using measurement-induced disturbance to characterize correlations as
  classical or quantum.
\newblock Phys.\ Rev.\ A \textbf{77}, 022301 (2008)

\bibitem{Z02}
Ollivier, H., Zurek, W.H.: Quantum discord: A measure of the quantumness of
  correlations.
\newblock Phys.\ Rev.\ Lett. \textbf{88}, 017901 (2001)

\bibitem{O02}
Oppenheim, J., Horodecki, M., Horodecki, P., Horodecki, R.: Thermodynamical
  approach to quantifying quantum correlations.
\newblock Phys.\ Rev.\ Lett. \textbf{89}, 180402 (2002)

\bibitem{PS08}
Pankowski, {\L{}}., Synak-Radtke, B.: Can quantum correlations be completely
  quantum?
\newblock J.\ Phys.\ A: Math.\ Theor. \textbf{41}, 075308 (2008)

\bibitem{P96}
Peres, A.: Separability criterion for density matrices.
\newblock Phys.\ Rev.\ Lett. \textbf{77}, 1413--1415 (1996)

\bibitem{PHH08}
Piani, M., Horodecki, P., Horodecki, R.: No-local-broadcasting theorem for
  multipartite quantum correlations.
\newblock Phys.\ Rev.\ Lett. \textbf{100}, 090502 (2008)

\bibitem{neg:i}
Plenio, M.B.: Logarithmic negativity: A full entanglement monotone that is not
  convex.
\newblock Phys.\ Rev.\ Lett. \textbf{95}, 090503 (2005)

\bibitem{PV06}
Plenio, M.B., Virmani, S.: An introduction to entanglement measures.
\newblock Quantum Inf.\ Comput. \textbf{Vol. 7}, 1--51 (2007)

\bibitem{RS10}
Rahimi, R., SaiToh, A.: Single-experiment-detectable nonclassical correlation
  witness.
\newblock Phys.\ Rev.\ A \textbf{82}, 022314 (2010)

\bibitem{SRN08-2}
SaiToh, A., Rahimi, R., Nakahara, M.: Evaluating measures of nonclassical
  correlation in a multipartite quantum system.
\newblock Int. J. Quant. Inf. \textbf{6 {\rm (Supp. 1)}}, 787--793 (2008)

\bibitem{ourarxivEnCE}
SaiToh, A., Rahimi, R., Nakahara, M.: Mathematical framework for detection and
  quantification of nonclassical correlation.
\newblock {\rm arXiv:0802.2263 (quant-ph)}  (2008)

\bibitem{SRN08}
SaiToh, A., Rahimi, R., Nakahara, M.: Nonclassical correlation in a
  multipartite quantum system: Two measures and evaluation.
\newblock Phys.\ Rev.\ A \textbf{77}, 052101 (2008)

\bibitem{W89}
Werner, R.F.: Quantum states with {Einstein-Podolsky-Rosen} correlations
  admitting a hidden-variable model.
\newblock Phys.\ Rev.\ A \textbf{40}, 4277--4281 (1989)

\end{thebibliography}

\end{document}